\documentclass[conference]{IEEEtran}
\usepackage{comment}
\usepackage{multicol}
\usepackage{blindtext}
\usepackage{amssymb}
\usepackage{subfigure}
\usepackage{amsmath}
\usepackage{amsfonts}
\usepackage{mathrsfs}
\usepackage{engord}
\usepackage[dvips]{graphicx}
\usepackage{bbm}
\usepackage{threeparttable}
\usepackage{booktabs}
\usepackage{epsfig}
\usepackage{color}
\usepackage{graphicx}
\usepackage{multirow}
\usepackage{cite}
\usepackage{epstopdf}
\usepackage{amsthm}
\usepackage[rawfloats=true]{floatrow} 
\restylefloat{figure} 

\usepackage{flushend}
\usepackage{float}
\usepackage{framed}
\usepackage{url}
\usepackage[ruled,norelsize]{algorithm2e}

\makeatletter
\newcommand{\removelatexerror}{\let\@latex@error\@gobble}
\makeatother
\begin{document}

\newtheorem{thm}{Theorem}
\newtheorem{prop}{Proposition}
\newtheorem{reproof}{Proof}
\newtheorem{lem}{Lemma}
\newtheorem{defn}{Definition}
\newtheorem{ex}{Example}
\newtheorem{cor}{Corollary}
\newtheorem{prn}{Principle}
\newtheorem{case}{Case}
%
% paper title
% can use linebreaks \\ within to get better formatting as desired
\title{Latent Factor Analysis of Gaussian Distributions under  Graphical Constraints}

\author{\IEEEauthorblockN{Md Mahmudul Hasan,
Shuangqing Wei, Ali Moharrer}}
\maketitle
\footnotetext[1]{Md M Hasan,   S. Wei and A. Moharrer  are with the school of Electrical Engineering and Computer Science, Louisiana State University, Baton Rouge, LA 70803, USA (Email: mhasa15@lsu.edu, swei@lsu.edu, alimoharrer@gmail.com). }

% author names and affiliations
% use a multiple column layout for up to three different
% affiliations

% make the title area

\begin{abstract}
We explore the algebraic structure of the  solution space of convex optimization problem Constrained Minimum Trace Factor Analysis (CMTFA), when  the population covariance matrix $\Sigma_x$ has an additional latent graphical constraint, namely, a latent star topology. In particular,  we have shown that CMTFA can have either a rank $ 1 $ or a rank $ n-1 $ solution and nothing in between. The special case of a rank $ 1 $ solution, corresponds to the case where just one latent variable captures all the dependencies among the observables, giving rise to a star topology. We found explicit conditions for both rank $ 1 $ and rank $n- 1$ solutions for CMTFA solution of $\Sigma_x$. As a basic attempt towards building a more general Gaussian tree, we have found a necessary and a sufficient condition for multiple clusters, each having rank $ 1 $ CMTFA solution, to satisfy a minimum probability to  combine together to build a Gaussian tree.   To support our analytical findings we have presented some numerical demonstrating the usefulness of the contributions of our work. 
\end{abstract}
\section{INTRODUCTION}
\subsection{Motivation}
Factor Analysis (FA) is a commonly used tool in multivariate statistics to represent the correlation structure of a set of observables in terms of significantly smaller number of variables called ``latent factors". With the growing use in data mining, high dimensional data analytics, factor analysis has already become a prolific area of research \cite{chen2017structured}\cite{bertsimas2017certifiably}. Classical factor analysis models seek to decompose the correlation matrix of an $n$-dimensional  random vector ${\bf X} \in {\mathcal R}^n$, $\Sigma_x  $, as the sum of a diagonal matrix $ D $ and a Gramian matrix $ \Sigma_{x}-D $. 

The literature that approached factor analysis can be classified in three major categories. Firstly, algebraic approaches \cite{albert1944matrices} and  \cite{drton2007algebraic}, where the principal aim was to give a characterization of the vanishing ideal of the set of symmetric $ n\times n $ matrices that decompose as the sum of a diagonal matrix and a low rank matrix, did not offer scalable algorithms for higher dimensional statistics. Secondly, factor analysis  via heuristic local optimization techniques, often based on the expectation maximization algorithm, were computationally tractable  but offered no provable performance guarantees. The third and final type of approach, based on convex optimization methods namely Minimum Trace Factor Analysis (MTFA)\cite{ledermann1940problem} and Minimum Rank Factor Analysis (MRFA)\cite{harman1976modern}, guaranteed performance and were computationally tractable. As the name suggests MRFA seeks to minimize the rank of $ \Sigma_{x}-D $ and MTFA minimizes the trace of $ \Sigma_{x}-D $. However,  MTFA solution could lead to negative values for the diagonal entries of the matrix $ D $. To solve this problem Constrained Minimum Trace Factor Analysis (CMTFA) was proposed \cite{bentler1980inequalities}, which imposes extra constraint of requiring $ D $ to be Gramian. Computational aspects of CMTFA and uniqueness of its solution were discussed in \cite{ten1981computational}. Though trace was used as the objective function, the paper set the groundworks for a broader class of convex optimization problems. Inspired by that groundwork Moharrer and Wei  in \cite{moharrer2017algebraic}  added another variety to the same class of problems that uses the determinant of a matrix as the objective function and named the problem Constrained Minimum Determinant Factor Analysis (CMDFA).

Gaussian graphical models \cite{gomez2014sensitivity} \cite{larranaga2013review} \cite{xiu2018multiple} have enjoyed  wide variety  of applications in  economics \cite{dobra2010modeling},  biology \cite{ahmed2008time} \cite{durbin1998biological},  image recognition \cite{besag1986statistical} \cite{geman1984stochastic}, social networks \cite{vega2007complex} \cite{wasserman1994social}  and many other fields. Among the Gaussian graphical models, we are particularly interested in the Gaussian latent tree models \cite{shiers2016correlation} where the ovservables are the leaves of the tree and the unovserved variables are the interior nodes. In the simplest form a Gaussian latent tree with just one node is a 'star'. Gaussian latent trees are highly favored because of their sparse structure \cite{mourad2013survey} and the availability of computationally efficient  algorithms to learn their underlying topologies \cite{choi2011learning} \cite{saitou1987neighbor}. For the purpose of clarity about the scope of our attention, it should be noted that, since there are algorithms available to learn the population covariance matrix $ \Sigma_{x} $ form data, that part of the work is beyond our scope and we are assuming $ \Sigma_{x} $ is known to us. Some of the existing algorithms addressed the problem of assigning a latent structure to the available data \cite{chow1968approximating}, we are not doing that either. Such algorithmic approaches designate sparse structures to the data for interpretation but they can not   guarantee  optimality in terms of the number of nodes, number of latent variables and the weight of individual edges of the graphical structure \cite{khajavi2018model} \cite{li2018learning}\cite{chow1968approximating}\cite{mourad2013survey}. Our primary focus is on the optimality and the solution space of the optimal solution. Hence, our scope is particularly limited to analyzing the solution space of the convex optimization problem CMTFA and the underlying structure to the data suggested by the solution.  

CMTFA guarrantees optimality of the solution, but to  designate a generative structure to the data we need the knowledge of the respective solution space . The main contribution of this work is based on the analysis of the solution space of CMTFA. Clearly our work is  not dealing  with the algorithm side of CMTFA, rather we aimed at understanding the solution space and associate a generative model for the data and thus make a more comprehensive and useful version of the respective solution. Though our ultimate goal is to associate a generative Gaussian tree model to the observed jointly Gaussian random vector $ \vec{X} $ ensuring the optimality of the solution, the scope of this particular version of our work is restricted to just one particular node of the tree. This is a handy first step towards building a complete and optimal tree structure. Figure \ref{figtree} exemplifies the aforementioned tree structure, where $ X_{i} $s are the observables forming the observable vector $ \vec{X} $ and $ Y_{j} $s are the latent variables  that capture the dependencies in $ X_{i} $s and form the vector of latent variables $ \vec{Y} $, whereas Figure \ref{figstar} demonstrates one particular node of a tree. 
\begin{figure}
\centering
\includegraphics[scale=0.5]{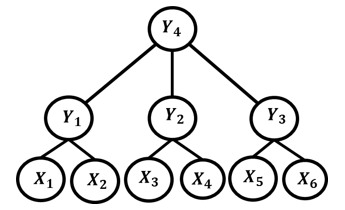}
\caption{Example of a Gaussian Tree}
\label{figtree}
\end{figure}
\begin{figure}
\centering
\includegraphics[scale=0.5]{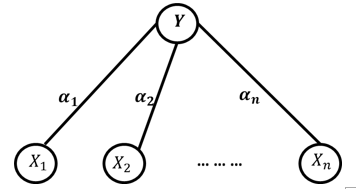}
\caption{A star topology}
\label{figstar}
\end{figure}

The motivation behind this work resides in the merit of trace as an  objective function. 
Ideally rank minimization approaches would lead to the least number of latent factors but they are computationally much more challenging than trace minimization approaches. Firstly, trace minimization is favored over rank minimization because of the fact that the trace of a matrix being a continuous function  offers more flexibility than the rank of matrix which is a discrete function. Secondly, Mitra and Alizahed proved in a dissertation \cite{fazel2002matrix} that, trace as an objective function is huristically as effective as the rank of a matrix. For example in \cite{liu2014factor}, in a tensor completion problem, nuclear norm of a matrix defined as the sum of the singular values of a tensor, was minimized as an equivalent problem to rank minimization. This is justified, because in \cite{fazel2002matrix} it was shown that  nuclear norm is the closest covex surrogate to the rank of a matrix. For the specific class of symmetric positive semi-definite  matrix that we are dealing with, the nuclear norm of the matrix is same as the sum of the eigen values of the matrix. Which makes trace of the matrix the closest convex surrogate to the rank of the matrix. So, in our search for the minimun number of latent factors to explain the origin of a set of observables, trace minimization is the closest feasible technique to rank minimization.  

The most closely related work to our work is  \cite{saunderson2012diagonal}, where the same necessary and sufficient condition was found on the subspace of $ \Sigma_{x} $ for MTFA that we found for CMTFA  solution of $ \Sigma_{x} $ to recover a star structure when $\Sigma_x$ is equipped with a latent star graphical constraint.   For clarification, the recovery of a star topology under a star constraint in a factor analysis problem means the resulting decomposition of $\Sigma_x$  ends up with $\Sigma_x-D$  having rank one i.e. a single latent variable can interpret the correlation entries in $\Sigma_x$. The main difference between their work \cite{saunderson2012diagonal} and ours is that, we also fully characterized the solution to CMTFA under a star constraint  for situations where the recovery of the latent star fails, an issue which they did not address. In particular, we proved that there are only two possible solutions to the CMTFA problem under a latent star constraint, one of which is the recovery of the star (i.e. the optimal number of latent variable is $k=1$), and the other with the optimal number of latent variables $k=n-1$. We found sufficient and necessary conditions for both cases. 

So the primary focus of this work is to explore the  solution space of a convex optimization problem, namely CMTFA, under a latent star constraint on $\Sigma_x$. Since our ultimate goal is to build a Gaussian tree and simultaneously maintain the locally optimum solutions, as a basic attempt towards that goal, we have found a necessary condition for multiple clusters, each having rank $ 1 $ CMTFA solution, to satisfy a minimum probability to  combine together to build a Gaussian tree. The insights obtained in this study will play a critical role when seeking analytical results of the factor analysis problems, when $\Sigma_x$ has more general latent tree structure, which is under investigation and will be presented in our future works. 
   
To support the rigorous analytical work that we have carried out,  at the end of this paper we have presened some numerical data. The numerical data shows  the difference between traces of the optimal solution to the star solution when the optimal solution is not a rank $ 1 $ matrix. In the next paragraph we enumerate the notable contributions of our work. 

 The major contributions of our work can be listed as the following,
\begin{itemize}
   \item We characterized the solution space of CMTFA. 
   \item Found necessary and sufficient conditions under which CMTFA solution of $ \Sigma_{x} $ is a star as well as when it is not a star. 
      \item We have both analytically and numerically shown the optimality of a non-star solution over the naive adoption of star when star is actually not the CMTFA solution of $ \Sigma_{x} $. 
   \item We have found a necessary and a sufficient condition for multiple clusters combine together to form a Gaussian tree, when each of the clusters has rank $ 1 $ CMTFA solution. 
  \end{itemize}

The rest of the paper is organized as follows: section II has the important definitions and notations, formulation of the problem is given in section III, necessary and sufficient condtions  for  rank $ 1 $ and rank $ n-1 $ CMTFA solutions are presented in sections  IV and V respectively, section VI has numerical results justifying the analytical findings of our work, section VII has the necessary and sufficient conditions for building a Gaussian tree under special settings, section VIII has the conclusion and at the end we have references.

\section{Definitions and Notations}
Let $ \vec{b} $ be a real $ n $ dimensional column vector  and $ A $ be an $ n\times n $ matrix. As in literature in general we denote the $ i $th element of $ \vec{b} $ as $ b_{i} $ and the $ (i,j) $th element of $ A $ as $ A_{i,j} $. Here we define all the vector operations and notations in terms of  $ \vec{b} $ and $ A $, that will carry their meaning on other vectors and  matrices throughout this paper unless stated otherwise. Along with these general notaions there will be some specific notations in the paper which we will define in the context of the particular cases they appear. 

 Vectors $ \vec{a}_{i,*} $ and $  \vec{a}_{*,i}  $ denote the $i$th row and $ i $th column vector of matrix $ A $ respectively. Function $ \lambda_{min}(A) $ is defined to be the smallest eigen value of matrix $ A $. $ N(A) $ stands for the null space of matrix $ A $. 

Vectors $ \vec{1} $ and $ \vec{0} $ are the $ n $ dimensional column vectors with each element equal to $ 1 $ and $ 0 $ respectively. When we write $ \vec{b}\geq 0 $ we mean that each element of the vector $ b(i)\geq 0, 1\leq i \leq n $.  $ \vec{b}^{2} $ is the Hadamard product of vector $ \vec{b} $ with itself. $ ||\vec{b} ||$ denotes the $ L_{2} $ norm of vector $ \vec{b} $. 

Now we define two terms i.e. \textit{dominance} and \textit{non-dominance} of a vector which will repeatedly appear throughout the paper. When we talk about the dominance or non-dominance of any vector $ \vec{b} $ we assume that the elements of the vector are sorted in a way such that $ |b_{1}|\geq |b_{2}|\geq \dots \geq |b_{n}|$. We call vector $ \vec{b} $ dominant and $ b_{1} $ the dominant element if for the above sorted vector $ |b_{1}|>\sum_{j\neq 1} |b_{j}| $  holds.  Otherwise $ \vec{b} $ is non-dominant.

\section{Formulation of the Problem} 
As we mentioned in the introduction, we are assuming that the population covariance matrix $ \Sigma_{x} $ is available, hence learning $ \Sigma_{x} $ from data is beyond our concern. We are not doing data mining here either i.e. building latent patterns to explain the origin of the data, as we are aware that those solutions may not be optimal ones. We are particularly concerned with analyzing the CMTFA solution space of $ \Sigma_{x} $ and understanding the underlying structures of the respective optimal solution. 

Traditional factor analysis problems seeks to decompose $ \Sigma_{x} $ as the sum of a low rank (rank $ <n $ ) component and a diagonal matrix. We are interested in the particular case where, the observables  $ \{X_{1}, ..., X_{n}\}$  are jointly Gaussian random variables, forming the jointly Gaussian random vector $ \vec{X}\sim \mathcal{N}(0,\Sigma_{x}) $  and all the diagonal entries of $ \Sigma_{x} $ are $ 1 $ as given by by \eqref{sigmax}.
 \begin{equation}\label{sigmax}
\Sigma_{x}=\begin{bmatrix}
1&\alpha_{1}\alpha_{2}&\dots&\alpha_{1}\alpha_{n}\\
\alpha_{2}\alpha_{1}&1&\dots&\alpha_{2}\alpha_{n}\\
\vdots&\vdots&\ddots&\vdots\\
\alpha_{n}\alpha_{1}&\alpha_{n}\alpha_{2}&\dots&1\\
\end{bmatrix}
\end{equation}
where  $ 0< |\alpha_{j}|< 1 $, $ j= 1,2, \dots ,n $ form the real column vector $ \vec{\alpha} $, $\vec{\alpha} = [\alpha_1, \dots, \alpha_n]' \in {\mathcal R}^n$  and 
\begin{align}\label{order_alpha}
|\alpha_{1}|\geq |\alpha_{2}| \geq \dots \geq |\alpha_{n}|  
\end{align}
The above $ \Sigma_{x} $ is  equipped with  a latent star topology, which graphically refers to Figure \ref{figstar}, which is  one particular node of the tree given by Figure \ref{figtree}. As we can see, a set of observables $ X_{1}, \dots , X_{n} $ connected to a common latent variable $ Y $ makes it a star topology indicating the  conditional independence given by \eqref{conditional independence} among the observables.

\begin{align}\label{conditional independence}
p(X_{1}X_{2},\dots, X_{n}|Y)=\Pi_{i=1}^{n}p(X_{i}|Y)
\end{align} 
The above conditional independence equivalently refers the generation of $ \Sigma_{x} $ to the following graphical model. 
\begin{align}
& \begin{bmatrix}
X_{1}\\
\vdots\\
X_{n}
\end{bmatrix}
=\begin{bmatrix}
\alpha_{1}\\
\vdots\\
\alpha_{n}
\end{bmatrix}
Y
+
\begin{bmatrix}
Z_{1}\\
\vdots\\
Z_{n}
\end{bmatrix}\label{graphical model}\\
\Rightarrow &\vec{X}= \vec{\alpha}Y+\vec{Z}
\end{align}
where
\begin{itemize}
\item $ \{X_{1}, ..., X_{n}\}$  are conditionally independent Gaussian random variables given $ Y $, forming the jointly Gaussian random vector $ \vec{X}\sim \mathcal{N}(0,\Sigma_{x}) $ where $ Y\sim \mathcal{N}(0,1) $.
\item $ \{Z_{1}, ..., Z_{n} \}$ are independent Gausian random varables with $Z_{j}\sim \mathcal{N}(0,1-\alpha_{j}^{2})\quad 1\leq j \leq n $ forming the Gaussian random vector $ \vec{Z} $.
\end{itemize}
The problem of finding a low rank solution for the decomposition of $ \Sigma_{x} $ under certain constrains has been equivalently formulated as a particular class of convex optimization problem in \cite{della1982minimum}, which ensures both optimality and structure unlike the algorithmic approaches. CMTFA used the trace of a matrix as the objective function and sought to decompose $ \Sigma_{x} $ as  given by \eqref{decomposition},
such that the trace of $ \Sigma_{t}$ (defined as $ \Sigma_{t}=\Sigma_{x}-D $) is minimized or equivalently the trace of $ D $ is maximized under the constraint that both $ (\Sigma_{x}-D) $ and $ D $ are Gramian matrices.  
\begin{align}\label{decomposition}
\Sigma_{x}=(\Sigma_{x}-D)+D
\end{align}

As CMTFA seeks low rank solutions to equation \eqref{decomposition}, the end result  may not necessarily be a rank $ 1 $ solution. When it actually ends up with  a rank $ 1 $  solution,  
we have just one latent variable producing $ n $ observables as in Figure \ref{figstar}, which corresponds to matrix $ \Sigma_{t} $ in  equation \eqref{decomposition} having the rank $ 1 $  solution given by,
\begin{align}\label{sigmatnd}
\Sigma_{t,ND}=\begin{bmatrix}
\alpha_{1}^{2}&\alpha_{1}\alpha_{2}&\dots&\alpha_{1}\alpha_{n}\\
\alpha_{2}\alpha_{1}&\alpha_{2}^{2}&\dots&\alpha_{2}\alpha_{n}\\
\vdots&\vdots&\ddots&\vdots\\
\alpha_{n}\alpha_{1}&\alpha_{n}\alpha_{2}&\dots&\alpha_{n}^{2}\\
\end{bmatrix}
\end{align}
The above rank $ 1 $ solution to CMTFA equivalently refers the generation of  $ \Sigma_{x} $ to the  graphical model given by \eqref{graphical model}.

However, since the solution to \eqref{decomposition} may not always be rank $ 1 $, it remains to be seen if CMTFA solution to $ \Sigma_{x} $ actually recovers the graphical model given by \eqref{graphical model}. Also to be investigated are the exact solutions to CMTFA problem, if it fails to recover the underlying star topology. 

So the problem that we are looking at, can be stated as follows: we are trying to find a close form analytical solution for CMTFA problem and gain insights about the underlying graphical structure. To be more specific our primary focus is to see if the underlying structure of  CMTFA solution to $ \Sigma_{x} $ with a star constraint is still a star or mathematically speaking to see  if $(\Sigma_{x}-D^{*}) $ is a rank one matrix given that $ D^{*} $ is the solution to \eqref{decomposition}. 
The following Theorem given in \cite{della1982minimum}, sets the ground rules for a matrix $ D^{*} $ to be the CMTFA solution for \eqref{decomposition}. 
\begin{thm}\label{necessary and sufficient}
The matrix  $  D^{*}$ is a solution of the CMTFA  problem if and only if $  D_{i,i}^{*}\geq 0, 1\leq i \leq n$, $ \lambda_{min}( \Sigma_{x}-D^{*})=0 $,  and there exists  $ n\times r $ matrix $ T $ such that $ \vec{t}_{*,i}\in \mathcal{N}(\Sigma_{x}-D^{*}), \quad i=1,....,r $ and the following holds,
\begin{align}\label{13}
\vec{1}=\sum_{i=1}^{r}\vec{t}_{*,i}^{2} - \sum_{j\in I( D^{*})}\mu_{j}\vec{\xi_{j}}
\end{align}
where $ r\leq n $ indicating the number of columns of the matrix $ T $,   $ I( D^{*})=\{i:D_{i,i}^{*}=0, 1\leq i\leq n\} $, $ \{\mu_{j}, \quad j\in I(D^{*}) \}$  are non-negative numbers and $\{ \vec{\xi}_{j}, j\in I(D^{*})\}$ are column vectors in $ {\mathcal R}^{n} $ with all the components equal to $ 0 $ except for the $ j $th component which is equal to $ 1 $.
\end{thm}
The theorem clearly specifies the requirements both for a matrix to be a candidate for the CMTFA solution of $ \Sigma_{x} $, as well as the null space matrix of that solution matrix.

In the next section, we explicitly analyze the conditions under which the CMTFA solution to $ \Sigma_{x} $  recovers the graphical model given by \eqref{graphical model} or equivalently speaking CMTFA solution becomes the rank $ 1 $ matrix given by \eqref{sigmatnd}. In section V we prove the sufficient and necessary condition that, when CMTFA solution of $ \Sigma_{x} $ does  not recover a star structure i.e., if the solution is not a rank $ 1 $ solution, then the solution is rank $ n-1 $. 
\section{CMTFA Non-dominant Case}
\begin{figure}
\centering
\includegraphics[scale=0.5]{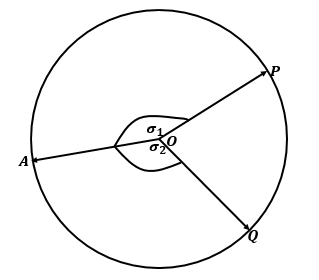}
\caption{Equal length on the surface of a sphere. (dimension $ n=3 $)}
\label{fig1sphere}
\end{figure}
In this section we analyze the conditions under which the CMTFA solution of $ \Sigma_{x} $ recovers a star structure. Understanding Lemma \ref{1lem} will be a good preparatory work before we proceed to state and prove Theorem \ref{1th}. The Lemma  also provides a geometric interpretaion that helps us view the problem in a broader perspective. 
\begin{lem}\label{1lem}
Non-dominance of vector $ \vec{\alpha} $ given by \eqref{non-dominanceof_alpha} is a necessary condition for the existence of such  $ n\times r $ matrix $ T $  that 
$ \vec{t}_{*,i}\in N(\Sigma_{t,ND}), \quad 1\leq i \leq r $  and $ ||\vec{t}_{j,*}||^{2}=1, \quad 1\leq j \leq n $. 

\begin{align}\label{non-dominanceof_alpha}
|\alpha_{1} |\leq \sum_{i=2}^{n}|\alpha_{i}|
\end{align}
\end{lem}
\begin{proof}[\textbf{Proof of Lemma \ref{1lem}:}]
Let $T $ be an $ n\times n $  matrix i.e  $ r=n $. We need,

\begin{align}
&\Sigma_{t,ND}T=\vec{0}\quad
\Rightarrow \vec{\alpha}T=\vec{0}\quad
\Rightarrow \sum_{i=1}^{n}\alpha_{i}\vec{t}_{i,*}=\vec{0}\notag\\
\Rightarrow &\alpha_{1}\vec{t}_{1,*}=-\sum_{i=2}^{n}\alpha_{i}\vec{t}_{i,*}\notag\\
\Rightarrow &||\alpha_{1}\vec{t}_{1,*}||^{2}=||-\sum_{i=2}^{n}\alpha_{i}\vec{t}_{i,*}||^{2}\notag\\
\Rightarrow &||\alpha_{1}\vec{t}_{1,*}||^{2}\leq\sum_{i=2}^{n}||\alpha_{i}\vec{t}_{i,*}||^{2}\\
\Rightarrow &|\alpha_{1}| ||\vec{t}_{1,*}||\leq\sum_{i=2}^{n}|\alpha_{i}| ||\vec{t}_{i,*}|| \notag\\
\Rightarrow & |\alpha_{1}|\leq \sum_{i=2}^{n}|\alpha_{i}|, \quad [\textit{because, $ ||\vec{t}_{i,*}||=||\vec{t}_{1,*}||, 2\leq i \leq n $}\notag
\end{align}
That completes the proof of the Lemma.
\end{proof}
For  a $ 3 $ dimensional geometric interpretation of the above necessary condition, let us consider that the matrix $ T $ has $ 3 $ row vectors.  Since we require  $  ||\vec{t}_{1,*}||^{2}=||\vec{t}_{2,*}||^{2}=||\vec{t}_{3,*}||^{2}=1 $, we can consider $ \vec{t}_{1,*}$, $\vec{t}_{2,*}$ and $\vec{t}_{2,*} $ to be three different points on the surface of a $ 3 $ dimensional hemisphere of radius $1 $ and be represented by the vectors $ \vec{OA}$, $\vec{OP}$ and  $\vec{OQ} $ respectively as in Figure \ref{fig1sphere}. Now, if $ |\alpha_{1}|> |\alpha_{2}|+|\alpha_{3}| $ it will be impossible to have  $ ||\alpha_{1}\vec{OA}||=|| \alpha_{2}\vec{OP}+\alpha_{3}\vec{OQ}|| $. But if $ |\alpha_{1}|\leq |\alpha_{2}|+|\alpha_{3}| $, we can always choose angles $\sigma_{1} $ and $\sigma_{2} $ such that $ ||\alpha_{1}\vec{OA}||=||\alpha_{2}\vec{OP}+\alpha_{3}\vec{OQ}|| $ holds, which is necessary for the orthogonality between the vector $ [\alpha_{1}, \alpha_{2}, \alpha_{3}]  $ and matrix $ T $ in this particular case. 

Now we  proceed to state and prove the statement of Theorem \ref{1th}, that has the main outcome of this subsection. 
\begin{thm}\label{1th}
CMTFA solution of $ \Sigma_{x} $ is $ \Sigma_{t,ND} $ if and only if $ \vec{\alpha} $ is non-dominant. 
\end{thm}
According to this theorem the CMTFA solution to a star connected network is a star itself,  if and only if   the  elements of  vector $ \vec{\alpha} $  satisfy equation \eqref{non-dominanceof_alpha}.

Before we move to the proof of the Theorem, it is worthwhile to mention that the statement of Theorem  \ref{1th} was proven in Theorem $ 3.4 $ of \cite{saunderson2012diagonal} for MTFA to recover a star structure, we prove the Theorem for CMTFA. In addition to that, we go onto find explicit condition and solution for the case when CMTFA does not recover a star structure as presented in the next section, which was not addressed in \cite{saunderson2012diagonal}. 
\begin{proof}[\textbf{Proof of Theorem \ref{1th}:}]
We recall the necessary and sufficient condition for CMTFA solution set by Theorem \ref{necessary and sufficient}. Since, $ \Sigma_{t,ND} $ in rank $ 1 $, its minimum eigenvalue is $ 0 $. To complete the proof of the Theorem, we only need to show the existance of rank $ n-1 $ matrix $ T$ such that the column vectors of $ T$ are in the null space of $ \Sigma_{t,ND} $ and the  $ L_{2} $-norm square of each row of $ T $ is $ 1 $. 

Lemma \ref{1lem} has already shown that, for the existence of such $ T $ non-dominance of vector $ \vec{\alpha}$ given by equation \eqref{non-dominanceof_alpha} is a necessary condition. Next we show that non-dominace is also a sufficient condition by constructing such a matrix $ T $ under the assumption of non-dominance of $ \vec{\alpha} $, and that should complete the proof of Theorem \ref{1th}. 

Its trivial to find  the following  basis vectors for the null space of  $ \Sigma_{t,ND} $,
\begin{align}
v_{1}=\begin{bmatrix}
-\frac{\alpha_{2}}{\alpha_{1}}\\
1\\
0\\
\vdots\\
0
\end{bmatrix}, v_{2}=\begin{bmatrix}
-\frac{\alpha_{3}}{\alpha_{1}}\\
0\\
1\\
\vdots\\
0
\end{bmatrix},\dots, \quad v_{n-1}=\begin{bmatrix}
-\frac{\alpha_{n}}{\alpha_{1}}\\
0\\
0\\
\vdots\\
1
\end{bmatrix}
\end{align}
We define matrix $ V $ so that the column vecors of $ V $ span the null space of $ \Sigma_{t,ND} $,
\begin{align}\label{v}
&V\notag\\
&= \begin{bmatrix}
-\frac{\alpha_{2}}{\alpha_{1}}&\dots&-\frac{\alpha_{n}}{\alpha_{1}}&-\left(c_{2}\frac{\alpha_{2}}{\alpha_{1}}+\dots+c_{n}\frac{\alpha_{n}}{\alpha_{1}}\right)\\
1&\dots&0&c_{2}\\
0&\dots&0&c_{3}\\
\vdots&\ddots&\vdots&\vdots\\
0&\dots&1&c_{n}
\end{bmatrix}
\end{align} 
where $c_{j}\in \{ 1, -1\}, \quad 2\leq j \leq n $. It will suffice for us to show the existance of  $ \{c_{j}\}, \quad 2\leq j \leq n $ and a diagonal matrix $ B $ such that the following holds. 
\begin{align}\label{1t}
T= V \cdot B
\end{align}
where, $ L_{2} $-norm square of each row of $ T $ is $ 1 $. Using \eqref{1t}, 
\begin{align}\label{1tt}
TT'= VBB'V'
\end{align}
We define the symmetric matrix $ \beta=BB' $, and  the diagonal  matrix $ \beta $ has only non-negative entries. Since we want each diagonal element of $ TT' $ to be $ 1 $, we have the following $ n $ equations,

\begin{align}\label{e1}
&\frac{\alpha_{2}^{2}}{\alpha_{1}^{2}}\beta_{11}+\frac{\alpha_{3}^{2}}{\alpha_{1}^{2}}\beta_{22}+\dots+\frac{\alpha_{n}^{2}}{\alpha_{1}^{2}}\beta_{n-1,n-1}+\notag\\
&\left(c_{2}\frac{\alpha_{2}}{\alpha_{1}}+c_{3}\frac{\alpha_{3}}{\alpha_{1}}+\dots+c_{n}\frac{\alpha_{n}}{\alpha_{1}}\right)^{2}\beta_{nn}=1
\end{align}
\begin{align}\label{e2}
\beta_{ii}+c_{i+1}^{2}\beta_{nn}=1, \quad i=1,\dots,n-1
\end{align}
Solving \eqref{e1}, we get, 
\begin{align}
&\beta_{nn}=
\frac{\alpha_{1}^{2}-\alpha_{2}^{2}-\alpha_{3}^{2}-\dots-\alpha_{n}^{2}}{\sum_{i\neq j, i\neq 1, j\neq 1}c_{i}c_{j}\alpha_{i}\alpha_{j}}
\end{align}
Equation \eqref{e2} indicates that, to ensure that all the diagonal entries of $ \beta $ are  non-negative, we need $ \beta_{nn}\leq 1 $. We select $ c_{i}, 2\leq i \leq n $ such that, 

\begin{align}
c_{i}\alpha_{i}=|\alpha_{i}|, \quad i=2,\dots, n
\end{align}

Under such selection of $ c_{i}, 2\leq i \leq n $, we have
\begin{align}
\beta_{nn}=\frac{\alpha_{1}^{2}-\alpha_{2}^{2}-\dots-\alpha_{n}^{2}}{\sum_{i\neq j, i\neq 1, j\neq 1}|\alpha_{i}||\alpha_{j}|}
\end{align}
Using the non-dominance assumption given in \eqref{non-dominanceof_alpha}, we have
\begin{align}
&\alpha_{1}^{2}\leq \left(\sum_{i=2}^{n}|\alpha_{i}|\right)^{2}\notag\\
\Rightarrow &\frac{\alpha_{1}^{2}-\sum_{i=2}^{n}\alpha_{i}^2}{\sum_{i\neq j, i\neq 1, j\neq 1}|\alpha_{i}||\alpha_{j}|}\leq 1\label{1bc}\\
\Rightarrow & \beta_{nn}\leq 1
\end{align}
Hence, non-dominance of vector $ \vec{\alpha} $ is a sufficient condition to construct the kind of $ T $ matrix required by a star structured CMTFA solution of $ \Sigma_{x} $. That completes the proof of Theorem \ref{1th}.
\end{proof}
\subsubsection*{CMTFA Boundary Case}
It is obvious that, there might be numerous ways to construct the matrix $ T $ that satisfy the requirements set by Theorem \ref{necessary and sufficient}.  Because of the special way we constructed the matrix $ T $, the rank of $ T $ under the non-dominant case is $ n-1 $ except for a very special case.  Here we talk about that special case of non-dominance i.e. when \eqref{non-dominanceof_alpha} holds for equality, then the rank of $ T $ is always $ 1 $ irrespective of the way we construct $T $ . For any given $ n $, it is straightforward to see from equation \eqref{1bc} that,  for  $ |\alpha_{1}|= \sum_{i=2}^{n}|\alpha_{i}| $ we have $ \beta_{nn}=1 $. Plugging $ \beta_{nn}=1 $ in equation \eqref{e2} gives us $ \beta_{ii}=0, 1\leq i \leq n-1$. Equations \eqref{1t} and \eqref{1tt} imply that, such a $ \beta $ matrix will produce a rank $ 1 $ matrix $ T $. This very special case is analytically explained by the next Lemma.
\begin{lem}\label{3lem}
When the non-dominance condition given in \eqref{non-dominanceof_alpha} holds for equality,
any $ n\times r $ matrix $ T $ such that 
$ \vec{t}_{*,i}\in N(\Sigma_{t,ND}), \quad 1\leq i \leq r $  and $ ||\vec{t}_{j,*}||^{2}=1, \quad 1\leq j \leq n $ has to be a rank $ 1 $ matrix. 
\end{lem}
\begin{proof}[\textbf{Proof of Lemma \ref{3lem}:}]
Using the orthogonality between $ \Sigma_{t,ND} $ and its null space matrix $ T $,
\begin{align}\label{ortho}
\sum_{i=1}^{n}\alpha_{i}\vec{t}_{i,*}=\vec{0}
\end{align}
Equation \eqref{ortho} implies the following two things:
\begin{align}
&||\alpha_{1}\vec{t}_{1,*}||=||\sum_{i=2}^{n}\alpha_{i}\vec{t}_{i,*}||\label{1imp}\\
&\alpha_{1}\vec{t}_{1,*}  =-\sum_{i=2}^{n}\alpha_{i}\vec{t}_{i,*}\label{2imp}
\end{align}
Using the triangular inequality,
\begin{align}\label{tri}
||\sum_{i=2}^{n}\alpha_{i}\vec{t}_{i,*}||&\leq \sum_{i=2}^{n}||\alpha_{i}\vec{t}_{i,*}||
\end{align}
We segregate the inequality in \eqref{tri} in two parts. The first part has,
\begin{align}
||\sum_{i=2}^{n}\alpha_{i}\vec{t}_{i,*}||&< \sum_{i=2}^{n}||\alpha_{i}\vec{t}_{i,*}||\notag\\
&=||\vec{t}_{1,*}||\sum_{i=2}^{n}|\alpha_{i}|=|\alpha_{1}|||\vec{t}_{1,*}||=||\alpha_{1}\vec{t}_{1,*}||\notag
\end{align}
which violates \eqref{1imp} hence orthogonality.
And the second part has,
\begin{align}
||\sum_{i=2}^{n}\alpha_{i}\vec{t}_{i,*}||&= \sum_{i=2}^{n}||\alpha_{i}\vec{t}_{i,*}||\notag\\
&=||\vec{t}_{1,*}||\sum_{i=2}^{n}|\alpha_{i}|=|\alpha_{1}|||\vec{t}_{1,*}||=||\alpha_{1}\vec{t}_{1,*}||\notag
\end{align}
which implies that all the $ \alpha_{i}\vec{t}_{i,*}, 1\leq i \leq n $ act  in the same line. Equivalently,  matrix $ T $ becomes a rank $ 1 $ matrix.
\end{proof}
\section{Dominant Case}
Having proved that CMTFA solution of $ \Sigma_{x} $ recovers a star structure only under the non-dominance of vector $ \vec{\alpha} $, in this section we explore  the CMTFA solution space  under the dominant case i.e.
\begin{align}\label{dominant_alpha}
|\alpha_{1}|>\sum_{i=2}^{n}|\alpha_{i}|
\end{align}
We show next that   under a dominant vector $ \vec{\alpha} $, CMTFA solution for $ \Sigma_{x} $ is a rank $ n-1 $ matrix given by \eqref{sigmatdm}. That means, CMTFA solution of $ \Sigma_{x} $ can either be rank $ 1 $ or rank $ n-1 $, nothing in between.  
 \begin{align}\label{sigmatdm}
&\Sigma_{t,DM}\notag\\
&=\begin{bmatrix}
(\Sigma_{t,DM})_{11}&\alpha_{1}\alpha_{2}&\dots&\alpha_{1}\alpha_{n}\\
\alpha_{2}\alpha_{1}&(\Sigma_{t,DM})_{22}&\dots&\alpha_{2}\alpha_{n}\\
\vdots&\vdots&\ddots&\vdots\\
\alpha_{n}\alpha_{1}&\alpha_{n}\alpha_{2}&\dots&(\Sigma_{t,DM})_{nn}\\
\end{bmatrix}
\end{align}
where
\begin{align}
&(\Sigma_{t,DM})_{11}=|\alpha_{1}|\left(\sum_{i\neq 1}|\alpha_{i}|\right)\notag\\
&(\Sigma_{t,DM})_{ii}=|\alpha_{i}|\left(|\alpha_{1}|-\sum_{j\neq i,1}|\alpha_{j}|\right),  i=2,\dots, n\notag
\end{align} 
Understanding the next two Lemmas will prepare us for the Theorem to follow.
\begin{lem}\label{l1}
$ \Sigma_{t,DM} $ is a rank $ n-1 $ matrix. 
\end{lem}
\begin{proof}[\textbf{Proof of Lemma \ref{l1}:}]
Let $ \gamma_{i}\in \{-1 , 1\} $ be the sign of $ \alpha_{i} $, i.e. $ \alpha_{i}=\gamma_{i}|\alpha_{i}| $.

For the $ 1 $st column of $ \Sigma_{t,DM} $,
\begin{align}
\sum_{g=2}^{n} \gamma_{1} \gamma_{g}(\Sigma_{t,DM})_{g1}&=\sum_{g=2}^{n} \gamma_{1} \gamma_{g} \gamma_{1} \gamma_{g}|\alpha_{g}||\alpha_{1}|\notag\\
&=\sum_{g=2}^{n} |\alpha_{g}||\alpha_{1}|\notag\\
&=|\alpha_{1}|\left(\sum_{g=2}^{n}|\alpha_{g}|\right)=(\Sigma_{t,DM})_{11}\notag
\end{align}
For the $ h $th $ (h\neq 1) $ column  of $ \Sigma_{t,DM} $,
\begin{align}
\sum_{g=2}^{n} \gamma_{1} \gamma_{g} (\Sigma_{t,DM})_{gh}
=&\gamma_{1} \gamma_{h}|\alpha_{1}||\alpha_{h}|-\sum_{m\neq h,1} \gamma_{1} \gamma_{h}|\alpha_{h}||\alpha_{m}|\notag\\
&+\sum_{m\neq h,1} \gamma_{1} \gamma_{m} \gamma_{m} \gamma_{h}|\alpha_{h}||\alpha_{m}|\notag\\
=&\gamma_{1} \gamma_{h}|\alpha_{1}||\alpha_{h}|-\sum_{m\neq h,1} \gamma_{1} \gamma_{h}|\alpha_{h}||\alpha_{m}|\notag\\
&+\sum_{m\neq h,1} \gamma_{1} \gamma_{h}|\alpha_{h}||\alpha_{m}|\notag\\
=&\gamma_{1} \gamma_{h}|\alpha_{1}||\alpha_{h}|=(\Sigma_{t,DM})_{1h}\notag
\end{align}
Combining the above two results,
\begin{align}
&(\Sigma_{t,DM})_{1,*}=\sum_{g=2}^{n} \gamma_{1} \gamma_{g}(\Sigma_{t,DM})_{g,*}\notag\\
& \Rightarrow(\Sigma_{t,DM})_{1,*}-\sum_{g=2}^{n} \gamma_{1} \gamma_{g}(\Sigma_{t,DM})_{g,*}=0\notag
\end{align}
Since the value of $ \gamma_{1} \gamma_{g} \in \{1,-1 \}, 2\leq g \leq n $, the above equation suggests that there exist nonzero coefficients $S_{g}\in \{1,-1\}  $ such that $\sum_{g=1}^{n}S_{g}(\Sigma_{t,DM})_{g,*}=0  $. Hence, we can conclude that the matrix $ \Sigma_{t,DM} $ is rank $ n-1 $. 
 \end{proof} 
\begin{lem}\label{l2}
There exists a column vector $\mathbf{\Phi}=[\Phi_{1}, \Phi_{2}, .... , \Phi_{n}]' $ such that $ \Sigma_{t,DM}\mathbf{\Phi}=0$, where $ \Phi_{i}\in \{-1, 1\}, 1\leq i \leq n $.
\end{lem}
This Lemma basically refers to the construction of the one dimensional null space required by Theorem \ref{necessary and sufficient} for a rank $ n-1 $ CMTFA solution of $ \Sigma_{x} $. 
\begin{proof}[\textbf{Proof of Lemma \ref{l2}:}]
It is obvious to see that the following selection of the elements of vector  $ \mathbf{\Phi} $ makes it orthogonal to $ (\Sigma_{t,DM})_{1} $, i.e. $ (\Sigma_{t,DM})_{1}\mathbf{\Phi}=0 $. Where $ (\Sigma_{t,DM})_{1} $ is the $ 1 $st row of $ \Sigma_{t,DM} $.
\begin{equation*}
    \Phi_{i}=
    \begin{cases}
      -1, & \alpha_{1} \alpha_{i}>0, i\neq 1\\
    1, & otherwise\\
        \end{cases}
  \end{equation*}
Now it will be sufficient to prove that any vector $ \mathbf{\Phi} $ orthogonal to $ (\Sigma_{t,DM})_{1} $ is also orthogonal to all the other rows of $ \Sigma_{t,DM} $, i.e. $ (\Sigma_{t,DM})_{i}\mathbf{\Phi}=0, 2\leq i\leq n $. 

Let $ \gamma_{i}\in \{-1 , 1\} $ be the sign of $ \alpha_{i} $, i.e. $ \alpha_{i}=\gamma_{i}|\alpha_{i}| $

Now for any row $ g $, $ g\neq 1 $,
\begin{align}
(\Sigma_{t,DM})_{g}\mathbf{\Phi}&=\Phi_{g}(\Sigma_{t,DM})_{gg}+\sum_{g\neq h} \Phi_{h}(\Sigma_{t,DM})_{gh}\notag\\
&=\Phi_{g} |\alpha_{g}|\left(|\alpha_{1}|-\sum_{i\neq g,1}|\alpha_{i}|\right)\notag\\
&+\sum_{g\neq h} \Phi_{h}\alpha_{g}\alpha_{h}\notag
\end{align}
\begin{align}
&=\Phi_{g} |\alpha_{g}||\alpha_{1}|+\Phi_{1} \alpha_{g}\alpha_{1}\notag\\
&-\sum_{i\neq g, i\neq 1} \Phi_{g}|\alpha_{g}||\alpha_{i}|+\sum_{h\neq g, h\neq 1} \Phi_{h}\alpha_{g}\alpha_{h}\notag\\
&=(\Phi_{g} +\Phi_{1}\gamma_{g}\gamma_{1})| \alpha_{g}||\alpha_{1}|\notag\\
&+\sum_{h\neq g, h\neq 1} (\gamma_{g}\gamma_{h}\Phi_{h}-\Phi_{g})|\alpha_{g}||\alpha_{h}|\label{3aa}
\end{align}

If $ \Phi_{g}=\Phi_{h} \Rightarrow \gamma_{1}\gamma_{g}=\gamma_{1}\gamma_{h} \Rightarrow \gamma_{g}= \gamma_{h}\Rightarrow \gamma_{g}\gamma_{h}\Phi_{h}-\Phi_{g}=0 $. \\
Else if $ \Phi_{g}\neq \Phi_{h} \Rightarrow \gamma_{1}\gamma_{g}\neq \gamma_{1}\gamma_{h} \Rightarrow \gamma_{g}\neq \gamma_{h}\Rightarrow \gamma_{g}\gamma_{h}\Phi_{h}-\Phi_{g}=0 $.\\
Similarly, \
If $ \Phi_{g}=\Phi_{1}\Rightarrow \alpha_{1}\alpha_{g}<0 \Rightarrow \gamma_{1}\neq \gamma_{g}\Rightarrow \Phi_{g} +\Phi_{1}\gamma_{g}\gamma_{1} =0$\\
Else if $ \Phi_{g}\neq \Phi_{1}\Rightarrow \alpha_{1}\alpha_{g}>0 \Rightarrow \gamma_{1}= \gamma_{g}\Rightarrow \Phi_{g} +\Phi_{1}\gamma_{g}\gamma_{1} =0$\\
Plugging these results in  equation \eqref{3aa}, we get
\begin{align}
(\Sigma_{t,DM})_{g}\mathbf{\Phi}=0\notag
\end{align}
And that completes the proof.
\end{proof}
Having proved the two Lemmas, we are now well equipped to state and prove  Theorem  \ref{2th}.
\begin{thm}\label{2th}
$ \Sigma_{t,DM} $ given by equation \eqref{sigmatdm} is the CMTFA solution of $ \Sigma_{x} $ if and only if $ \vec{\alpha} $ is dominant. 
\end{thm} 
\begin{proof}[\textbf{Proof of Theorem \ref{2th}}]
To prove the Theorem we refer to necessary and sufficient condition set by Theorem \ref{necessary and sufficient}. Lemma \ref{l1} proves that the
rank of $ \Sigma_{t,DM} $ is $ n-1 $, so its minimum eigenvalue is $\lambda_{min}(\Sigma_{t,DM})= 0 $. Since $0< |\alpha_{i}|<1 $ and $ 0<(\Sigma_{t,DM})_{ii})<1, i=1,\dots, n $, all the diagonal entries $ D_{i,i}^{*}, 1\leq i \leq n $ of the matrix $ D^{*} $ are positive. As a result, the set $I(D^{*})$ is empty and the second term in the right hand side of \eqref{13} vanishes.

The dimension of the null space of $ \Sigma_{t,DM} $ is $ 1 $. It will suffice for us to prove the existence of a  column vector $ \mathbf{\Phi}_{n\times 1} $, $\Phi_{i} \in \{1,-1\}, 1\leq i \leq n $ such that $ \Sigma_{t,DM}\mathbf{\Phi}=0 $. Lemma \ref{l2} gives that proof. 
\end{proof}
So, the outcome of the section is that when the CMTFA solution of $ \sigma_{x} $  does not end up with a rank $ 1 $ solution i.e. does not recover a star structure, it ends up with a rank $ n-1 $ solution and this happens if and only if $ \vec{\alpha} $ is dominant. 
\section{Numerical results}
In this section we present some numerical data along with their analytical insights to show how much of an advantage does the optimal solution for CMTFA dominant case  give over the naive adoption of star solution. It is straightforward to quantify such advantage using equations \eqref{sigmatnd} and \eqref{sigmatdm}, in terms the difference of trace between the two solutions. 
\begin{figure}
\centering
\includegraphics[scale=0.5]{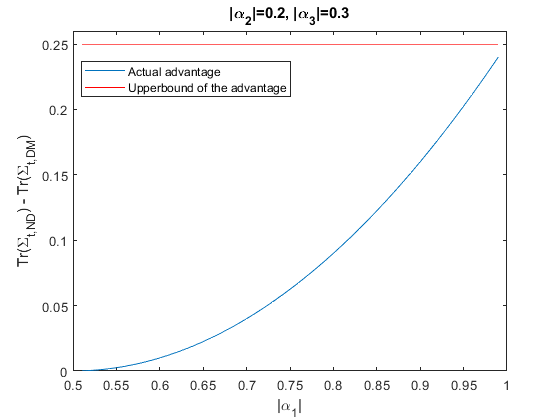}
\caption{Atvantage of the optimum solution over the the star solution  for the dominant case plotted against different values of $ |\alpha_{1}| $ for given $ |\alpha_{2}| $ and $ |\alpha_{3}| $  }
\label{figtd}
\end{figure}
\begin{align}\label{trace difference}
\mathrm{Tr}(\Sigma_{t,ND})-\mathrm{Tr}(\Sigma_{t,DM})&=|\alpha_{1}|\left(|\alpha_{1}|-2\sum_{i=2}^{n}|\alpha_{i}|\right)\notag\\
&+\left(\sum_{i=2}^{n}|\alpha_{i}|\right)^{2}
\end{align}
where the function $ \mathrm{Tr}(.) $ denotes the trace of a matrix. Without the loss of generality we can vary the dominance of $ \vec{\alpha} $ by increasing $ |\alpha_{1}|$ while keeping $ |\alpha_{j}|, 2\leq j \leq n $ unchanged. Under such settings the terms $ 2\sum_{i=2}^{n}|\alpha_{i}| $ and $ \left(\sum_{i=2}^{n}|\alpha_{i}|\right)^{2} $ are constants, which in turn suggests from equation \eqref{trace difference} that, $ \mathrm{Tr}(\Sigma_{t,ND})-\mathrm{Tr}(\Sigma_{t,DM}) $ is an increasing function of $ |\alpha_{1}| $ as shown in Figure \ref{figtd}. Since $ |\alpha_{1}|<1 $, the term $ \mathrm{Tr}(\Sigma_{t,ND})-\mathrm{Tr}(\Sigma_{t,DM}) $ must be upperbounded by $ 1-2\sum_{i=2}^{n}|\alpha_{i}|
+\left(\sum_{i=2}^{n}|\alpha_{i}|\right)^{2} $.

\section{Building a Gaussian Tree}
Though the major focus of this work is just one node of the tree, in this section we present some of our findings towards combining multiple clusters to build a Gaussian tree. 

We are assuming  that we have a $\Sigma_{x}$ generated with a latent Gaussian tree topology constrainst and these Gaussian observable variables are further divided into $m$ clusters. The $i$th cluster has $n_{i}$ observables i.e. $X_{i,1}, \cdots, X_{i, n_{i}}$. The individual edge weights of the  observable variables $ \alpha_{j}$  are iid uniformly over $ (-1,1) $ or equivalently  $|\alpha_i|$ are iid uniformly over $ (0,1) $ and the correlation between any two observable $X_i$ and $X_j$ is the product of the respective edge weights of the path connecting the two observed variables in the latent Gaussian tree \cite{moharrer2017information}. 

Under the above settings, the problem at our hand can be stated as follows: we have $m$ randomly generated clusters equipped with a Gaussian latent tree structural interpretation and we require the probability of all of them non-dominant in order to keep the $m$ local clusters recontructable at least equal to $\delta$. 

The next Lemma will help us prove the Theorem to follow. 

\begin{lem}\label{lemma}
The probability that a particular cluster with $ n $ observables is non-dominant or equivalently the vector $ \vec{\alpha}=[\alpha_{1}, \alpha_{2}, \dots \alpha_{n}]'  $ is non-dominant, is $1- \frac{1}{n!} $. 
\end{lem}
The Lemma essentially  gives the probability with which a cluster is non-dominant. It is trivial to show that multiple clusters can be combined to build a Gaussian tree when each of the clusters has only one latent variable i.e. each cluster is non dominant. 

Before we embark on the proof of the Lemma it should be noted that, though for the ease of analysis we assumed by equation \eqref{order_alpha} that $ |\alpha_{1}| $ is has the biggest absolute value of all the elements and hence only $ \alpha_{1} $ can be the dominant element, but in practice each element in $ \vec{\alpha} $ has equal probability to be the biggest value. So, to determine the probability with which $ \vec{\alpha} $ is dominant, we must consider the fact that each element has equal probability to be the dominant element. 

\begin{proof}
Let, $ S_{i} $ be defined as the sum of the absolute values of all the elements of vector $ \vec{\alpha} $ except the $i $th element as given by equation \eqref{si}. 
\begin{align}\label{si}
S_{i}=\sum_{j\neq i}|\alpha_{j}|, \quad 1 \leq i \leq n
\end{align}
 Before we embark on the main part of the proof, we need to find the probability density function of  $ S_{i} $ i.e. $ f_{S_{i}}(t) $  for  $ 0<t<1 $. Since We know that each $|\alpha_{i}|, 1\leq i \leq n$ has a probability density function of  $\mathcal{U}(0,1)  $, hence if $ n=3 $ then $ f_{S_{i}}(t) $ should be the convolution between two uniform density functions i.e. convolution between  $\mathcal{U}(0,1)$ and $\mathcal{U}(0,1)$, which would result in $  f_{S_{i}}(t) =t, \quad 0<t<1 $. If we now add one more variable i.e. make $ n=4 $, then the $ f_{S_{i}}(t) $ that we got from the $ 3 $ variable case will have to be convoluted with another $\mathcal{U}(0,1)$ to give us  the new $ f_{S_{i}}(t) $ for $ 4 $ variable case, which will result in $  f_{S_{i}}(t) =\frac{t^{2}}{2}, \quad 0<t<1 $. It is straightforward to see that, following similar steps will result in $ f_{S_{i}}(t)=\frac{t^{n-2}}{(n-2)!}, \quad 0<t<1 $ for $ n $ varialble case i.e. $ \vec{\alpha} $ being a column vector of dimension $ n $. 
 
 Let $h=|\alpha_{1}|$, then we have
\begin{align}
P(\vec{\alpha} \quad \text{is dominant})&=\sum_{i=1}^{n}\frac{1}{n}P(\textit{$ \alpha_{i} $ is the dominant element})\notag\\
&= \frac{1}{n}\sum_{i=1}^{n} \int_{0}^{1} P[(S_{i}<h) \mid \alpha_{i}] \quad dh\notag\\
&= \frac{1}{n}\sum_{i=1}^{n} \int_{0}^{1} \left[\int_{0}^{h}f_{S_{i}}(t) \quad  dt\right] \quad dh\notag\\
&= \frac{1}{n}\sum_{i=1}^{n} \int_{0}^{1} \left[\int_{0}^{h}\frac{t^{n-2}}{(n-2)!} \quad  dt\right] \quad dh\notag
\end{align}
\begin{align}
&= \frac{1}{n}\sum_{i=1}^{n} \int_{0}^{1} \left[\frac{h^{n-1}}{(n-1)!} \right] \quad dh\notag\\
&= \frac{1}{n}\sum_{i=1}^{n}  \left[\frac{1^{n}}{n!} \right]\notag\\
&= \frac{1}{n}n  \left[\frac{1}{n!} \right]\notag\\
&= \frac{1}{n!}\notag
\end{align}
So, the probability that $ \vec{\alpha} $ is non- dominant is
\begin{align}
P(\vec{\alpha} \quad \text{is non-dominant})&=1-P(\vec{\alpha} \quad \text{is dominant})\notag\\
&=1-\frac{1}{n!}\notag
\end{align}
The next Theorem has the main outcome of the section.
\end{proof} 
 \begin{thm}\label{theorem}
Equations \eqref{nc} and \eqref{sc} are respectively necessary and sufficient conditions for all of the $m$ randomly generated clusters, equipped with a Gaussian latent tree structural interpretation, to be non-dominant with a   probability at least equal to $\delta$.
\begin{align}\label{nc}
\frac{1}{m}\sum_{i=1}^{m}n_{i}!\geq\frac{1}{1-\delta^{\frac{1}{m}}}
\end{align}
\begin{align}\label{sc}
n_{min}!\geq \frac{1}{1-\delta^{\frac{1}{m}}}
\end{align}
where $n_{min}!=\min_{i}\{n_{1}!, \dots, n_{m}!\}$.
\end{thm}
The statement of the theorem is intuitively satisfying, because any increament in the value of $ m $ and decreament in the value of $ n $ will impact the chance of having all the clusters non-dominant negatively and vice versa. The necessary condition given by \eqref{nc} perfectly captures such tussle. And for the sufficient condition, if the minimum value in a set of positive numbers satisfies a certain lowerbound, then the other values in the set will automatically satisfy that lowerbound. 
\begin{proof}
Let the function $ f(n_{i}!) $ denote the probability that the $ i $th cluster is non-dominant. We have to satisfy the minimum probability $ \delta $ that all of the $ m $ clusters together form a Gaussian tree, which is possible if and only if all the clusters are non-dominant. Mathematically speaking, since the clusters are independent of each other, we have to satisfy the following equation.
\begin{align}\label{req}
\pi_{i=1}^{m}f(n_{i}!)\geq \delta
\end{align} 
or equivalently,
\begin{align}\label{ineq}
\frac{1}{m}\sum_{i=1}^{m}\ln f(n_{i}!)\geq \frac{\ln \delta}{m}
\end{align} 
Now, because of the concavity of the $ \ln(.) $ function, the following must be true,
\begin{align}\label{ineq1}
\frac{1}{m}\sum_{i=1}^{m}\ln f(n_{i}!)\leq \ln \left[ \frac{ \sum_{i=1}^{m} f(n_{i}!)}{m}\right]
\end{align}
We define, $ x=n_{i}! $, and using Lemma \ref{lemma} gives us $ f(x)=1-\frac{1}{x} $ and in turn we get a negative value for the double derivative of $ f(x) $, $f"(x)=-\frac{2}{x^{3}} $. Using such concavity of the function $ f(.) $ and the inequality given by \eqref{ineq1}, we get
\begin{align}\label{ineq2}
\frac{1}{m}\sum_{i=1}^{m}\ln f(n_{i}!)\leq \ln \left[ f \left(\frac{ \sum_{i=1}^{m}n_{i}!}{m}\right)\right]
\end{align}
Combining, \eqref{ineq} and \eqref{ineq2} we get,
\begin{align}
&\ln \left[ f \left(\frac{ \sum_{i=1}^{m}n_{i}!}{m}\right)\right]\geq \frac{1}{m}\ln \delta\notag\\
\Longleftrightarrow & f \left(\frac{ \sum_{i=1}^{m}n_{i}!}{m}\right)\geq \delta^{\frac{1}{m}}\notag\\
\Longleftrightarrow & 1-\frac{1}{\frac{1}{m}\sum_{i=1}^{m}n_{i}!}\geq \delta^{\frac{1}{m}} \quad \text{[using Lemma \ref{lemma}]}\notag\\
\Longleftrightarrow & \frac{1}{m}\sum_{i=1}^{m}n_{i}!\geq \frac{1}{1-\delta^{\frac{1}{m}}}
\end{align}
That proves the necessary condition. Now to prove the sufficient condition, showing that when \eqref{sc} is true \eqref{req} holds, will suffice. 

\begin{align}
&n_{min}!\geq \frac{1}{1-\delta^{\frac{1}{m}}}\notag\\
\Rightarrow & 1-\frac{1}{n_{min}!}\geq \delta^{\frac{1}{m}} \notag\\
\Rightarrow &f\left(n_{min}!\right)\geq \delta^{\frac{1}{m}}\notag\\
\Rightarrow & \ln \left[f\left(n_{min}!\right) \right] \geq  \frac{1}{m} \ln \delta \notag\\
\Rightarrow & m \ln \left[f\left(n_{min}!\right) \right] \geq \ln \delta \notag\\
\Rightarrow & \sum_{i=1}^{m} \ln \left[f\left(n_{min}!\right) \right] \geq \ln \delta \notag
\end{align}
\begin{align}
\Rightarrow &  \ln \left[ \pi_{i=1}^{m} f\left(n_{min}!\right) \right] \geq \ln \delta \notag\\
\Rightarrow &   \pi_{i=1}^{m} f\left(n_{min}!\right)  \geq  \delta \notag\\
\Rightarrow &   \pi_{i=1}^{m} f\left(n_{i}!\right)  \geq  \delta \notag
\end{align}
\end{proof}
 \section{Conclusion}
In this paper we analyzed the solution space of   convex optimization problem CMTFA when the matrix $ \Sigma_{x} $ has a given latent-star interpretation. In particular we proved that the CMTFA solution of $ \Sigma_{x} $ is either rank $ 1 $ or rank $ n-1 $ and nothing in between. We characterized the closed form and found explicit sufficient and necessary conditions for both the solutions.  As a basic attempt towards building a more general Gaussian tree, we also found a necessary and a sufficient condition for multiple clusters having rank $ 1 $ CMTFA solution, to satisfy a minimum probability to  combine together to build a Gaussian tree.

\bibliography{arxiv_ref}
\bibliographystyle{IEEEtran}
\end{document}